%% file: main.tex
\documentclass[sigconf, 10pt]{acmart}

\usepackage{balance}
\usepackage{amsmath}
\usepackage{algorithm}
\usepackage{algorithmicx}
\usepackage[noend]{algpseudocode}
\usepackage{color}

\setcopyright{none} %comment to print copyright info
\renewcommand\footnotetextcopyrightpermission[1]{} % removes footnote with conference information in first column
\makeatletter
\def\BState{\State\hskip-\ALG@thistlm}
\makeatother

\newtheorem{theorem}{Theorem}[section]

\newtheorem{lemma}[theorem]{Lemma}

\begin{document}

\title{Privacy-Preserving Synthetic Datasets \\ Over Weakly Constrained Domains}

%\numberofauthors{3}
%\author{
 % \alignauthor{Luke Rodriguez\\
 %   \affaddr{University of Washington}\\
 %   \email{rodriglr@uw.edu}}
 % \alignauthor{Bill Howe\\
 %   \affaddr{University of Washington}\\
 %   \email{billhowe@uw.edu}}
 % \alignauthor{Julia Stoyanovich\\
%    \affaddr{Drexel University}\\
%    \email{stoyanovich@drexel.edu}}
%}
%\maketitle

\author{Luke Rodriguez}
\affiliation{
	\institution{University of Washington}
    \city{Seattle}
    \state{WA}
}
\email{rodriglr@uw.edu}

\author{Bill Howe}
\affiliation{
	\institution{University of Washington}
    \city{Seattle}
    \state{WA}
}
\email{billhowe@uw.edu}

% \author{Julia Stoyanovich}
% \affiliation{
% 	\institution{Drexel University}
% 	\city{Philadelphia}
%     \state{PA}
% }
% \email{stoyanovich@drexel.edu}

\begin{abstract}

\input{short_version/abstract}

\end{abstract}

% <concept_id>10002978.10003029.10011703</concept_id>
% <concept_desc>Security and privacy~Usability in security and privacy</concept_desc>
% <concept_id>10002978.10003029.10011150</concept_id>
% <concept_desc>Security and privacy~Privacy protections</concept_desc>
% <concept_significance>300</concept_significance>

% \ccsdesc[500]{Security and privacy~Usability in security and privacy}
% \ccsdesc[300]{Security and privacy~Privacy protections}

% \keywords{SOME KEYWORDS}

\maketitle

\section{Introduction}
\input{short_version/intro}

\section{Bad Bins Theorem}
\input{short_version/theorem}

\section{Algorithm}
\input{short_version/algorithm}

\section{Experiments}
\input{short_version/experiments}

\section{Conclusion}
\input{short_version/conclusion}

%SHORT VERSION TO HERE

\balance

\bibliographystyle{abbrv}
\bibliography{refs}

% That's all folks!
\end{document}

%% file: short_version/abstract.tex
Techniques to deliver privacy-preserving synthetic datasets take a sensitive dataset as input and produce a similar dataset as output while maintaining differential privacy.  These approaches have the potential to improve data sharing and reuse, but they must be accessible to non-experts and tolerant of realistic data.  Existing approaches make an implicit assumption that the active domain of the dataset is similar to the global domain, potentially violating differential privacy.

In this paper, we present an algorithm for generating differentially private synthetic data over the large, weakly constrained domains we find in realistic open data situations.  Our algorithm models the unrepresented domain analytically as a probability distribution to adjust the output and compute noise, avoiding the need to compute  the full domain explicitly.  We formulate the tradeoff between privacy and utility in terms of a ``tolerance for randomness'' parameter that does not require users to inspect the data to set.  Finally, we show that the algorithm produces sensible results on real datasets.

%% file: short_version/intro.tex
We consider differential privacy mechanisms to produce privacy-preserving synthetic datasets from real datasets without making any assumptions about the domain from which the original dataset is drawn. Our motivation is open data sharing situations, where the data publisher is typically not the data owner and therefore cannot make assertions about how the data was collected.

Consider, for example, a survey that prompts users to enter their \texttt{gender} in a free-text field. After collecting a sufficient number of responses, the data collector passes the to the data publisher for privatization and release, who notices that every response is either \texttt{male} or \texttt{female}. She models the distribution appropriately using one of several existing methods \cite{hay2010boosting,xiao2011differential,xiao2012dpcube}, and creates synthetic data with \texttt{gender} values of \texttt{male} and \texttt{female} with zero probability of any other response.  But a single respondent entering \texttt{genderqueer} as their response violates the differential privacy.  
%Even though a human may take responsiblity for asserti asserted that the domain is
% Maybe something about shifting responsibility to the human.  A human that follows a mechanistic procedure still violates privacy, even though it is not the "fault" of the software.  So we want to design practical systems that discourage risky decisions.

One solution is to ask the data provider to constrain the domain, but this approach changes the nature of the survey.  Another solution is to refuse to release data in this setting, citing the inapplicability of differential privacy.  But this failure seems unsatisfying, as it would typically be acceptable to just ignore the rare values and still release something useful.  In this paper, we consider the problem of deciding on an appropriate threshold for rare values, which turns out to be a subtle issue.

%This problem is not unique to categorical variables. Consider a database of political donations with schema (\texttt{donor}, \texttt{amount}, \texttt{campaign}). These donations can be as little as one cent, but there is no maximum theoretical amount. If the active domain of the dataset contained only donations below $\$1,000$, the synthetic data generated using any proposed methods for numeric histogram publication \cite{acs2012differentially,rastogi2010differentially,xu2013differentially} would also have $\$1,000$ as an upper limit. However, if a new version of the dataset included a contribution of $\$100,000$ from a wealthy donor, the bounds would change, exposing the fact that the wealthy donor was included in the dataset and violating her privacy.

The approach we take involves sampling from differentially private histograms; a survey by Meng et al summarizes the many solutions to this problem \cite{meng2017different}. However, when producing synthetic data, the histogram cannot be modeled as a parallel set of differentially private count queries, one for each bin.  Any possible count query over the whole domain could potentially be asked by a user, so we essentially need an extremely large number of bins. We show that any mechanism must include bins from outside the active domain with non-zero probability to preserve differential privacy (Theorem \ref{noactive}).

Our solution is to allow a (potentially large) domain for the attribute being modeled, and replace low-count bins with ``random'' bins with some probability.  We explore the relationship between the domain size and this probability of including unrealistically ``random'' values and domain size to empower the publisher to tailor their mechanism \emph{without inspecting the data}.  
%The publisher can set conservative (loose) bounds for the specified domain when the true domain is uncertain, or set tighter bounds when more information is available.  When the specified domain is too conservative, the synthetic dataset may include unrealistically extreme values, and utility will suffer.  But in all cases, differential privacy is preserved --- unlike existing approaches with utility results that rely on tight domain assumptions.  

We do not disallow the data publisher from making a conservative guess about the true domain under this approach.  A poor guess will only impact utility, not privacy. For example, a data publisher who sees an attribute labeled gender may simply assert that the domain only includes \texttt{male} and \texttt{female}.  Privacy is preserved, and utility may be acceptable.  However, this assertion would exclude any other values, significantly impacting utility. On the other hand, the data publisher may decide that any English word might be valid. This decision would allow for synthetic data outside of the male/female binary, but the size of the domain means that there is more likely to be considerable noise negatively impacting utility.  The only thing we want to discourage is mechanizing your ``guess'' as just equivalent to the active domain.

In the remainder of this short paper, we prove a simple theorem that motivates our approach, then formulate an algorithm that is a) efficient for large domains and b) allows control over the probability of including values from outside the active domain.  We then show that this algorithm produces sensible results on real data with no data-dependent input from the user.

%% file: short_version/theorem.tex
\label{sec:theorem}

We show the impossibility of producing DP synthetic data that includes only values from the active domain.  We begin with a Lemma:

\begin{lemma}
\label{zeroprob}
Given a differentially private mechanism 
$q$ with respect to a set of database instances $\mathcal{D}$ of schema $R(\mathcal{A})$, if $P[q(D_i) = r] = 0$ for some database instance $D_i$ and a result $r$, then $P[q(D) = r] = 0$ for all $D \in \mathcal{D}$.
\end{lemma}
\begin{proof}
Let $P[q(D_i) = r] = 0$ for some database instance $D_i$ and result $r$. For any database instance $D_j \in \mathcal{D}$, define a sequence of database instances $S(D_j, D_i) = \{D_1, D_2, \ldots , D_n\}$ such that $D_1 = D_i$, $D_n = D_j$, and $D_k$ and $D_{k+1}$ are neighbors for all $0 \leq k < n$. \footnote{Two datasets are neighbors if they differ in the presence or absence of a single record, following the differential privacy definition.} It follows that
\begin{align*}
	P[q(D_j) = r] &\leq e^{\epsilon} \times P[q(D_2) = r] \\
    &\leq e^{2\epsilon} \times P[q(D_3) = r] \\
    &\vdots \\
    &\leq e^{(n-1)\epsilon} \times P[q(D_i) = r] \\
    &\leq 0
\end{align*}
Since the probability cannot be less than $0$, we find that $P[q(D) = r] = 0$ for all $D \in \mathcal{D}$
\end{proof}

\begin{theorem}
\label{noactive}
There does not exist a differentially private mechanism that \emph{only} returns elements from the active domain.
\end{theorem}

Intuition: There must be a non-zero chance of returning every element in the global domain, or else Lemma \ref{zeroprob} is violated.

\begin{proof}
Let $\mathcal{D}$ be the set of possible database instances. Let $c$ be a function such that $c(D)$ returns the active domain of $D$. Let $c^*$ be some differentially private mechanism that returns a value set $c^*(D) \subseteq c(D)$.  This means that $c^*$ can be probabilistic, but only the value sets in the powerset of $c(D)$ have non-zero probabilities. Next let us pick some $D_i \in \mathcal{D}$. Then $P[q(D_i) = r] = 0$ for all $r \not\subseteq c(D_i)$, and additionally by Lemma \ref{zeroprob} that $P[q(D) = r] = 0$ for all $D \in \mathcal{D}$. This holds true for all $c(D)$, and thus the possible range $R$ of $c^*(D)$ is constrained by $R = \cap_{D \in \mathcal{D}} c(D) = \emptyset$. Therefore, our mechanism $c^*$ must be the trivial one that returns the empty set.
\end{proof}

Examining the proof of Theorem \ref{noactive}, we note that it depends crucially on the choice of function $c$. If we allow this function to be defined on a pre-defined domain outside of the active domain found in the database instance $D$, this range constraint no longer holds and the mechanism can in fact be differentially private. This observation is the motivation for the categorical method outlined in the next section. 

%% file: short_version/algorithm.tex
\label{sec:algorithm}

In this section we establish a differentially private algorithm for categorical (discrete) data given a constrained domain. The algorithm allows the data publisher to act unilaterally to produce a synthetic dataset with strong guarantees about privacy even when only very weak assumptions can be made about the domain, and even when these assumptions may be wrong. 

To illustrate the ideas, we will refer to the case of a single string attribute $A$ with a small active domain $adom(A)$ but with a global domain $dom(A)$ of all possible alphanumeric strings.  Our task is to select the bins we will include in the histogram.  Once we have the bins, we can add Laplace noise to their counts or use any other appropriate mechanism; we are not concerned with that step.

% \begin{figure}
% 	\centering
%     \includegraphics[width=\columnwidth]{figs/categorical.png}
% 	\caption{With a large domain and a naive solution, almost half of the global domain could be included in a noisy histogram.}
% 	\label{fig:cat_scenario}
% \end{figure}

\begin{figure}
	\centering
    \includegraphics[width=\columnwidth]{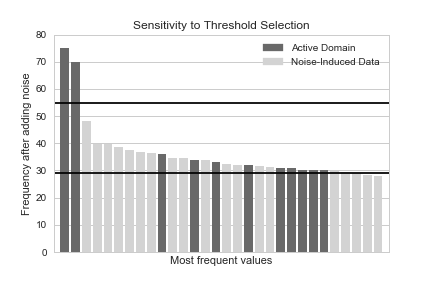}
	\caption{Illustration of how the true values of a long-tailed distribution (dark bars) can be obscured or excluded by the noise-induced bins of Algorithm \ref{cat_hist} (light bars).}
	\label{fig:cat_thresholds}
\end{figure}

One way to select the bins for the histogram is to use the exponential mechanism: sort the bins by their counts, and select $k$ bins with exponentially decaying probability proportional to their count.  However, repeated execution of this mechanism requires repeatedly splitting the privacy budget, and would require specifying $k$.  But setting $k$ appropriately requires inspecting the active domain.

Another approach is to use the Laplace mechanism by adding noise to every bin in the global domain.  Besides being computationally infeasible, this approach would cause utility to vanish. About half of the categories outside of the active domain would receive a positive weight due to the symmetry of Laplace noise, producing an enormous number of seemingly random bins. For a sufficiently large domain, the probability density of these random bins would outweigh the probability density of the categories found in the active domain. 

We could address this issue by only including those categories with noisy counts that exceed a threshold $\tau$. The larger this threshold is, the less likely we are to include any given category from outside the active domain.  But setting $\tau$ appropriately  requires inspecting the active domain: If it is set too high, rare values will not appear in the synthetic dataset.  If it is set too low, there will be a high number of bins with random labels. 

Instead, we express $\tau$ in terms of a tolerance for random values, a parameter which could potentially be set as a global policy across all datasets.  

%We know that the probability of not including any categories from outside the active domain can be viewed as a trial for each of the $n$ categories in the global domain not in the active domain, where each has probability $p$ of success, or equivalently that the noisy count of the category exceeds the threshold $\tau$.  

%\textbf{Controlling for random values:} We cannot select $\tau$ based on any inspection of the support of the active domain without incurring an additional cost to the privacy budget. Since the impact of $\tau$ on the user experience manifests in terms of increased likelihood of unexpected values in the histogram, we reformulate $\tau$ in terms of a tolerance for including values outside the active domain. This provides a parameter that can be applied across multiple datasets.
%To set this parameter, the publisher is permitted to iteratively generate synthetic datasets, inspect the results, and adjust the tolerance, as long as they only release the final dataset.

To express this tolerance, we model each bin outside of the active domain as a binary trial, where each bin is included with probability $p$ equal to the probability that the Laplace noise exceeds the threshold $\tau$. Assuming that $\tau \geq 0$, this constraint means that $p = \frac{1}{2}e^{-\epsilon\tau}$. Therefore we can sample from this binomial distribution over $n$ trials to find the number of categories to include from outside the active domain, where $n$ is the size of the global domain. Let $B(k,n,p)$ represent the Binomial probability of $k$ successes on $n$ trials with probability $p$. This expression overestimates the probability slightly, as the categories in the active domain are already accounted for. Using this relationship we can define a probability $\rho$ of including zero categories from outside the active domain: $\rho = B(0,n,p) = (1 - \frac{1}{2}e^{-\epsilon\tau})^n$.

We can set $\rho$ to any appropriate value and solve for the noisy threshold $\tau$ in terms of the other parameters, resulting in the relation $\tau = -\frac{1}{\epsilon}ln[2(1-\rho^{\frac{1}{n}})]$. Because this expression is derived from the assumption that $\tau \geq 0$, the relationship only holds when $\rho^{\frac{1}{n}} \geq \frac{1}{2}$. However, for any sizable value of $n$ we note that $\rho^{\frac{1}{n}} \gg \frac{1}{2}$ so this relationship holds in general.

By viewing this process as binomial selection, we also make the problem computationally feasible.  We can analytically determine the number of bins $k$ from outside the active domain that will exceed the threshold $\tau$.  Let $k = B(n,p)$ be a random variable chosen from the binomial distribution with $n$ and $p$ defined as before.  Then we randomly sample $k$ categories from the global domain and give each a weight sampled from the following probability distribution derived from Bayes' Theorem: $P(x | x \geq \tau) = \frac{P(x)}{P(x \geq \tau} = \epsilon e^{-\epsilon(x-\tau)}$, which is simply an exponential distribution with $\lambda = \epsilon$ right-shifted to $\tau$. Modeling the inactive domain with a probability distribution significantly reduces the computation necessary.

This procedure has the effect of replacing the original distribution present in the data with a version of an exponential distribution as a result of the nature of the Laplace noise being added to each of the values from outside of the active domain. This exponential distribution depends only on the size of the domain and not the data in the database instance $D$.  For high values of $\tau$, values in the tail of the distribution may be excluded.  But even for low values of $\tau$, it becomes likely that random bins will have higher counts than true bins, obscuring the true values.  Figure \ref{fig:cat_thresholds} illustrates the situation. %In general, the more you know about the population distribution of the data, the better you can do, but this mechanism applied to bin selection naturally obscures rare values and always achieves differential privacy.
\vspace{-1em}
\begin{algorithm}
  \caption{Categorical Histogram Method}\label{cat_hist}
  \begin{algorithmic}[1]
    \Require{Size of the domain $n$, tolerance for values outside the active domain $\rho$, privacy budget $\epsilon$, and the true histogram ($C$, $S$) where $C$ is a vector of categories in the active domain and $S$ is a vector of true frequencies for each corresponding category in $C$.}
    \Ensure{Differentially private histogram}
    \Procedure{CatHist}{$n, \rho, \epsilon, C, S$}
      \State $\tau \gets -\large\frac{\ln(2(1-\rho^{\frac{1}{n}}))}{\epsilon}$
      \State $i \gets 0$
      \While{$i < |C|$}
        \State $s_i \gets Lap(s_i, \frac{1}{\epsilon})$
      	\If{$s_i < \tau$} Remove($c_i, s_i$) \EndIf
        \State $i \gets i+1$
      \EndWhile
      \State $k \gets \frac{1}{2}e^{-\epsilon \tau}$ \Comment{Noisily select number of bins}
      \State $j \gets 0$
      \While{$j < k$}
      	\State Append($C$, GetCategoryFromDomain())
        \State Append($S$, $\tau + ExpDist(\epsilon)$)
      	\State $j \gets j+1$
      \EndWhile
      \Return{$C,S$}
      \EndProcedure
  \end{algorithmic}
\end{algorithm}

Algorithm \ref{cat_hist} shows the full procedure for categorical histograms. The algorithm accepts the size of the domain $n$, a tolerance for values outside the active domain $\rho$, a privacy budget $\epsilon$ and the true histogram $(C, S)$ where $C$ is a vector of categories in the active domain and $S$ is a vector of true frequencies for $C$. The proof that Algorithm \ref{cat_hist} is differentially private follows from the definition of the Laplace mechanism, and is omitted for space.

This algorithm has two parameters that can be varied to control the tradeoff of privacy and pollution from the large domain: $\rho$ and $\epsilon$. The $\rho$ parameter bounds the probability of including random values, and a low value increases the utility of the dataset for collaborators. However, a high $\rho$ will increase $\tau$, and therefore require a dataset whose categories have high support to produce useful results.

%% file: short_version/experiments.tex
\label{sec:experiments}

\begin{figure*}
	\centering
    \includegraphics[width=2.3in]{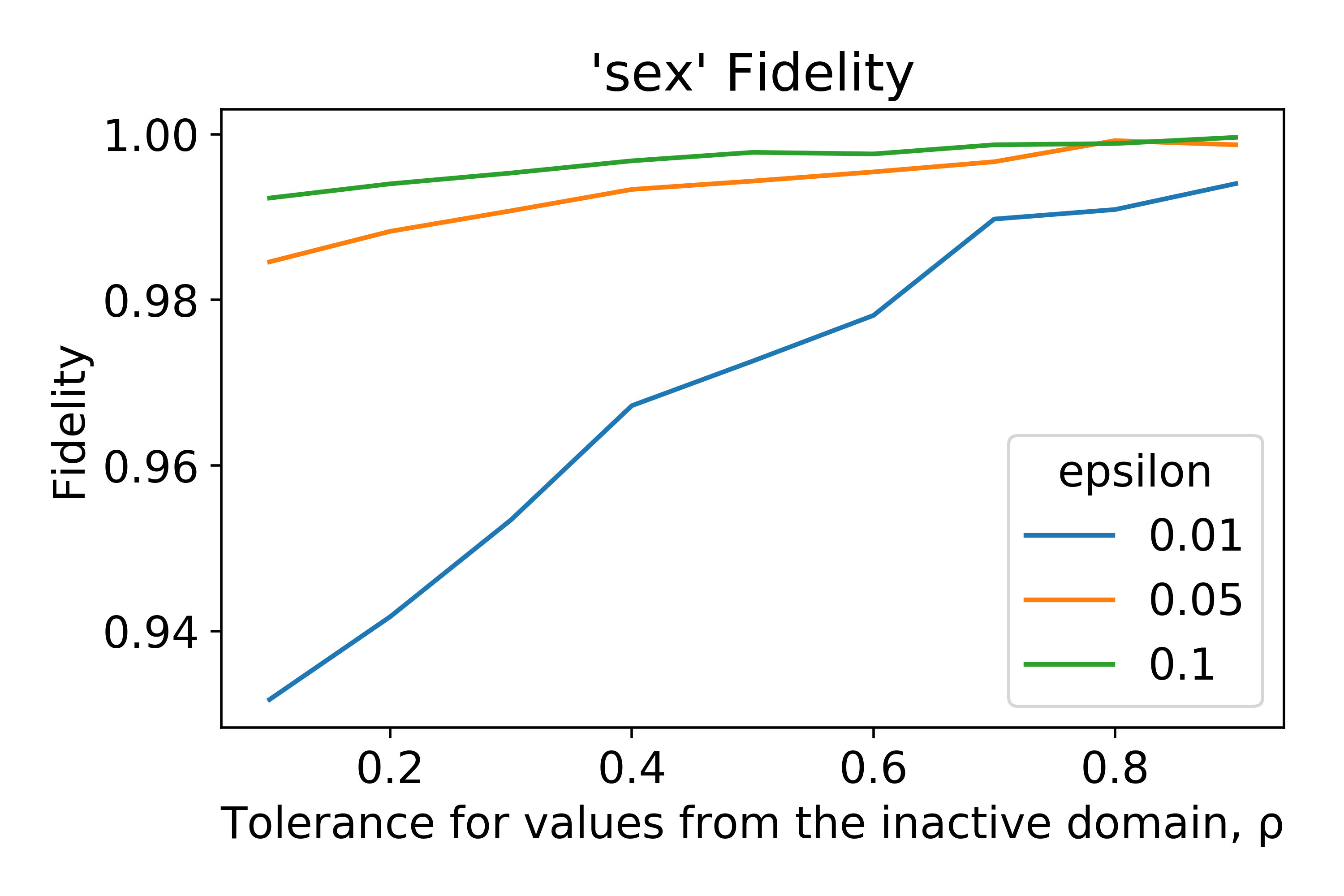}
    \includegraphics[width=2.3in]{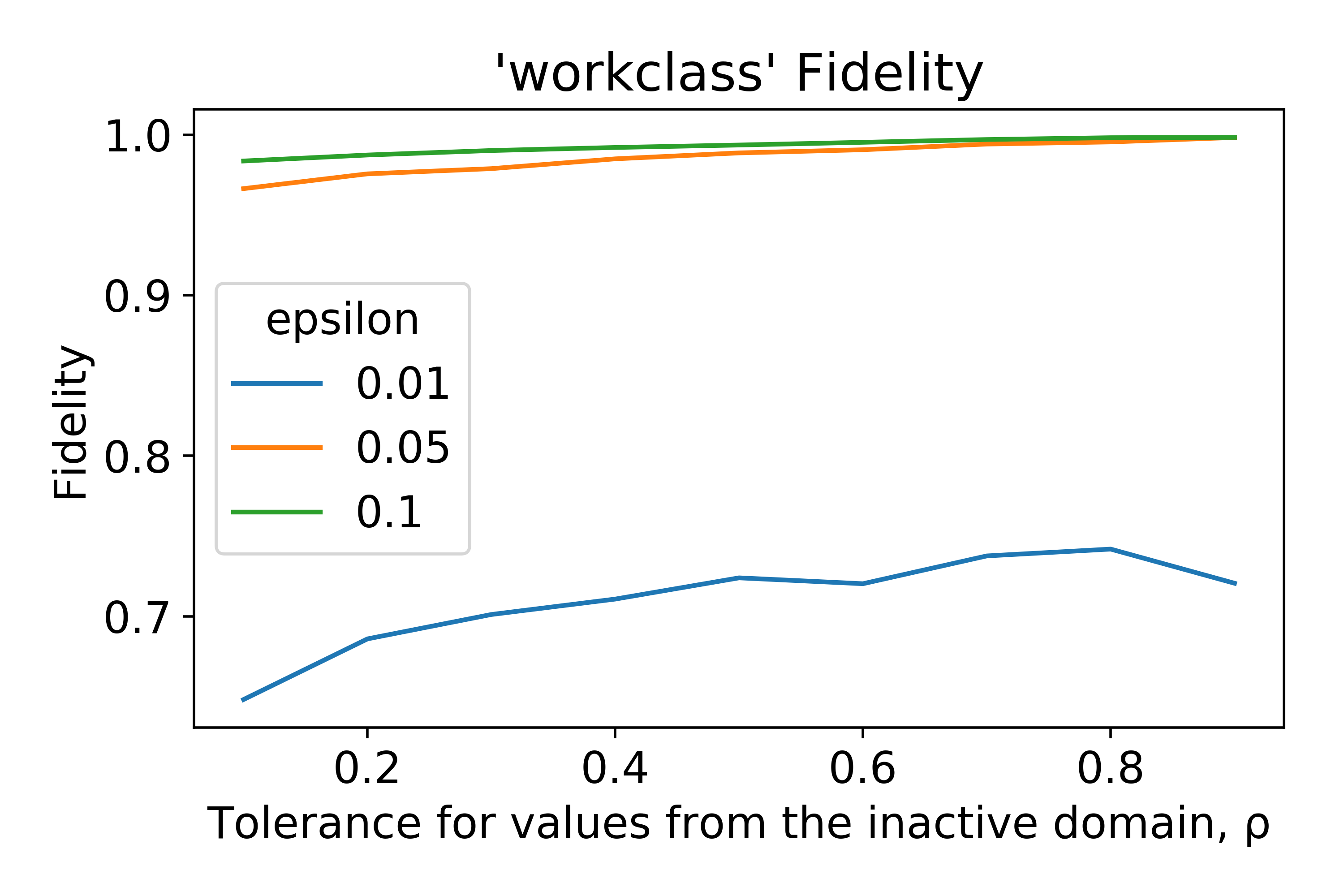}
    \includegraphics[width=2.3in]{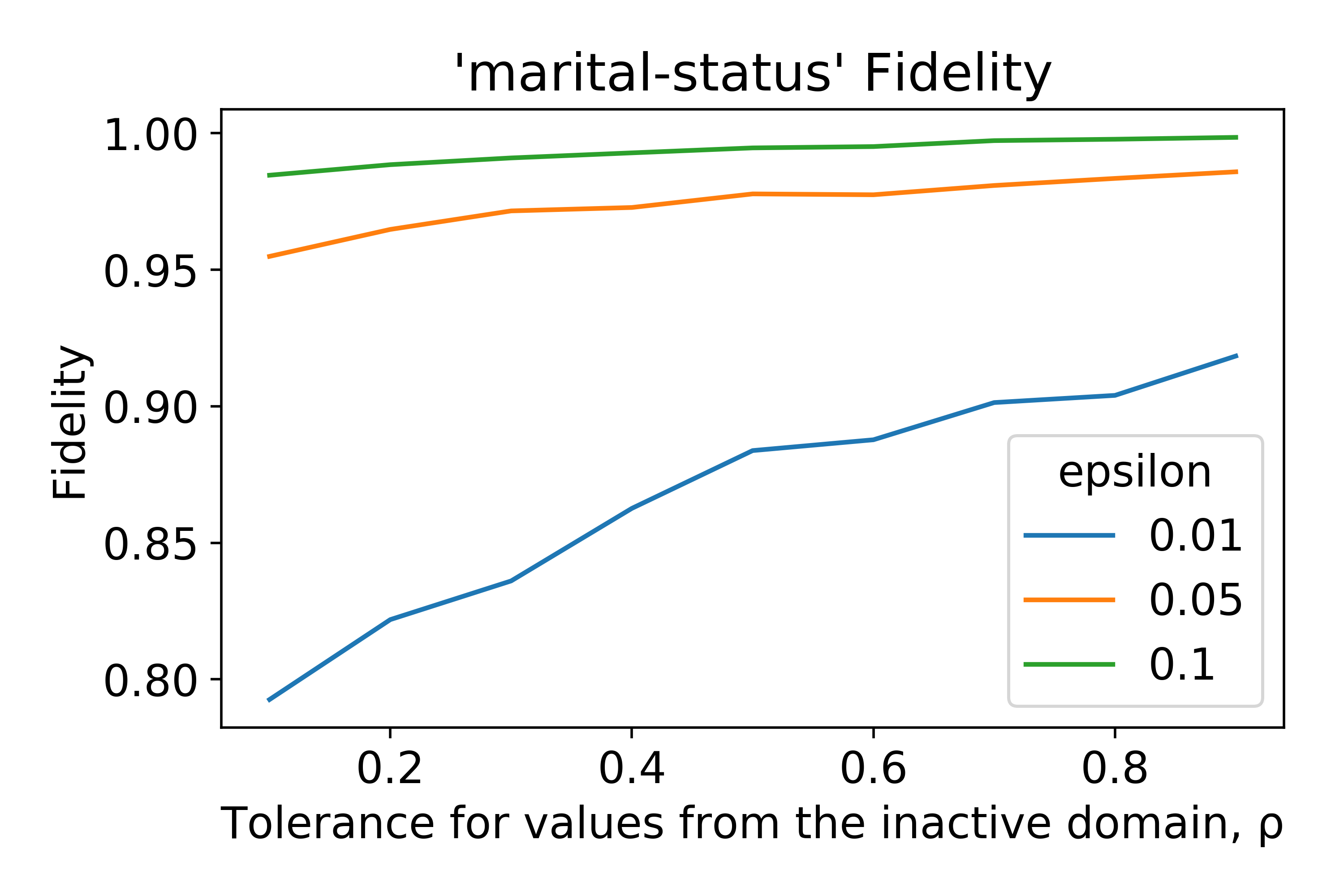}
	\caption{Fidelity results for categorical data over varying $\rho$ and $\epsilon$ values}
	\label{fig:cat_utility}
\end{figure*}
In order to evaluate the behavior of Algorithm \ref{cat_hist}, we run it on attributes from the UCI Adult dataset \cite{Lichman:2013}. We specifically investigate whether or not CatHist generates sensible results.

% \begin{figure*}
% 	\centering
%     \includegraphics[width=2in]{figs/sex.png}
%     \includegraphics[width=2in]{figs/workclass.png}
%     \includegraphics[width=2in]{figs/marital-status.png}
% 	\caption{True distributions for \texttt{sex}, \texttt{workclass}, and \texttt{marital-status}}
% 	\label{fig:cat_distributions}
% \end{figure*}

\textbf{Evaluation Metric:} A natural choice of metric is to evaluate the accuracy of some randomized query on the output as compared to that of the original data. However, random queries over the full domain would be mostly zero in both the true data and the synthetic data.  Instead we measure the fidelity of the results with respect to the categories present in the active domain. We define this metric $F$ to be the product of the probability densities of the true distribution and the results of Algorithm \ref{cat_hist} over the intersection of the category sets. Thus $F = 1$ if the category sets are equivalent, and $F = 0$ if the intersection of the two is the empty set. Note that approaches that include all of the categories in the active domain will have a fidelity $F = 1$ by definition (but suffer privacy concerns).

\textbf{Experiment:} The \texttt{sex}, \texttt{workclass}, and \texttt{marital-status} fields from the Adult dataset are investigated in our experiment. The domain of sex was taken to be all English words (size $1.71 \times 10^5$), while \texttt{workclass} and \texttt{marital-status} were run over all English word pairs (size $2.94 \times 10^{10}$). Algorithm 1 returns a category set $C$ altered from the active domain.

We ran Algorithm 1 over each of these categories 100 times for varying values of $\epsilon$ and $\rho$, and report the average $F$ score for each in Figure 6. As a general trend, we find that increasing $\rho$ increases the fidelity of the result, except for the \texttt{workclass} attribute when $\epsilon = 0.01$. This increase is a result of the fact that \texttt{workclass} is the longest-tailed distribution of the three we are observing, and for high $\rho$ values we are setting the threshold prohibitively high and obscuring most of the results in this tail. For higher $\epsilon$ values, the threshold is set low enough to include categories from this long tail, eliminating the behavior.

Another observation is that the fidelity for \texttt{sex} is consistently higher than that of the other two attributes. This increase is a result of the fact that the counts are distributed between two high-support attributes and that the specified domain size is much smaller. We can therefore choose thresholds well below the true support for the less common attribute value, leading to near-$100\%$ preservation of the attribute structure for $\rho = 0.9$.

%% file: short_version/conclusion.tex
In this paper we have shown the difficulty of meeting the constraints of differential privacy as a data publisher when the domain of data values is unknown, and presented a mechanism that allows for these constraints to be met when generating synthetic data from a histogram. The data publisher is free to assert any domain they wish (as long as it is not simply the active domain) and the results published using our method will be differentially private.  Utility suffers if the global domain is too large or a subset of the active domain. If the choice of global domain is similar to the active domain, then our approach converges to that of traditional methods.

Our approach adds an additional parameter to control the sensitivity to random values, but this extra parameter does not depend on the data. The privacy burden is therefore shifted from the data publisher to the system itself: privacy violations are impossible even if the data publisher has no knowledge of data privacy and no knowledge of the input data.  Although utility can be poor for large domains or low noise tolerances, we find that a synthetic dataset is always at least minimally useful to communicate the structure and types of the data to developers debugging new tools. Looking forward, approaches to handle weakly constrained domains can lead to tools that are easier to use for novices to facilitate data sharing.

%% file: main.bbl
\begin{thebibliography}{1}

\bibitem{hay2010boosting}
M.~Hay, V.~Rastogi, G.~Miklau, and D.~Suciu.
\newblock Boosting the accuracy of differentially private histograms through
  consistency.
\newblock {\em Proceedings of the VLDB Endowment}, 3(1-2):1021--1032, 2010.

\bibitem{Lichman:2013}
M.~Lichman.
\newblock {UCI} machine learning repository, 2013.

\bibitem{meng2017different}
X.~Meng, H.~Li, and J.~Cui.
\newblock Different strategies for differentially private histogram
  publication.
\newblock {\em Journal of Communications and Information Networks},
  2(3):68--77, 2017.

\bibitem{xiao2011differential}
X.~Xiao, G.~Wang, and J.~Gehrke.
\newblock Differential privacy via wavelet transforms.
\newblock {\em IEEE Transactions on Knowledge and Data Engineering},
  23(8):1200--1214, 2011.

\bibitem{xiao2012dpcube}
Y.~Xiao, L.~Xiong, L.~Fan, and S.~Goryczka.
\newblock Dpcube: differentially private histogram release through
  multidimensional partitioning.
\newblock {\em arXiv preprint arXiv:1202.5358}, 2012.

\end{thebibliography}
